\documentclass[pra,twocolumn,showpacs,preprintnumbers,amsmath,amssymb,superscriptaddress]{revtex4}

\usepackage{bm}
\usepackage{graphicx}
\usepackage{amsbsy}
\usepackage{amsmath}
\usepackage{amsfonts}
\usepackage{amsthm}

\begin{document}

\theoremstyle{plain}
\newtheorem{theorem}{Theorem}
\newtheorem{lemma}[theorem]{Lemma}
\newtheorem{corollary}[theorem]{Corollary}
\newtheorem{conjecture}[theorem]{Conjecture}
\newtheorem{proposition}[theorem]{Proposition}
\newcommand{\PT}{\mathrm{PTL}}
 \newcommand{\C}{\mathbb{C}}
  \newcommand{\F}{\mathbb{F}}
  \newcommand{\N}{\mathbb{N}}
  \renewcommand{\P}{\mathbb{P}}
  \newcommand{\R}{\mathbb{R}}
  \newcommand{\Z}{\mathbf{Z}}
  \renewcommand{\a}{\mathbf{a}}
  \renewcommand{\b}{\mathbf{b}}
  \renewcommand{\i}{\mathbf{i}}
  \renewcommand{\j}{\mathbf{j}}
  \renewcommand{\c}{\mathbf{c}}
  \newcommand{\e}{\mathbf{e}}
  \newcommand{\f}{\mathbf{f}}
  \newcommand{\g}{\mathbf{g}}
  \newcommand{\gl}{\mathbf{GL}}
  \newcommand{\m}{\mathbf{m}}
  \newcommand{\n}{\mathbf{n}}
  \newcommand{\bNP}{\mathbf{NP}}
  \newcommand{\bNPC}{\mathbf{NPC}}
  \newcommand{\p}{\mathbf{p}}
  \newcommand{\bP}{\mathbb{P}}
  \newcommand{\bPo}{\mathbf{Po}}
  \newcommand{\q}{\mathbf{q}}
  \newcommand{\s}{\mathbf{s}}
  \newcommand{\bt}{\mathbf{t}}
  \newcommand{\T}{\mathbf{T}}
  \newcommand{\U}{\mathbf{U}}
  \renewcommand{\u}{\mathbf{u}}
  \renewcommand{\v}{\mathbf{v}}
  \newcommand{\V}{\mathbf{V}}
  \newcommand{\w}{\mathbf{w}}
  \newcommand{\W}{\mathbf{W}}
  \newcommand{\x}{\mathbf{x}}
  \newcommand{\X}{\mathbf{X}}
  \newcommand{\y}{\mathbf{y}}
  \newcommand{\Y}{\mathbf{Y}}
  \newcommand{\z}{\mathbf{z}}
  \newcommand{\0}{\mathbf{0}}
  \newcommand{\1}{\mathbf{1}}
  \newcommand{\Gam}{\mathbf{\Gamma}}
  \newcommand{\bGamma}{\Gam}
  \newcommand{\Lam}{\mathbf{\Lambda}}
  \newcommand{\lam}{\mbox{\boldmath{$\lambda$}}}
  \newcommand{\bA}{\mathbf{A}}
  \newcommand{\bB}{\mathbf{B}}
  \newcommand{\bC}{\mathbf{C}}
  \newcommand{\bH}{\mathbf{H}}
  \newcommand{\bL}{\mathbf{L}}
  \newcommand{\bM}{\mathbf{M}}
  \newcommand{\bc}{\mathbf{c}}
  \newcommand{\cA}{\mathcal{A}}
  \newcommand{\cB}{\mathcal{B}}
  \newcommand{\cC}{\mathcal{C}}
  \newcommand{\cD}{\mathcal{D}}
  \newcommand{\cE}{\mathcal{E}}
  \newcommand{\cF}{\mathcal{F}}
  \newcommand{\cG}{\mathcal{G}}
  \newcommand{\cH}{\mathcal{H}}
  \newcommand{\cI}{\mathcal{I}}
  \newcommand{\cL}{\mathcal{L}}
  \newcommand{\cM}{\mathcal{M}}
  \newcommand{\cO}{\mathcal{O}}
  \newcommand{\cP}{\mathcal{P}}
  \newcommand{\cR}{\mathcal{R}}
  \newcommand{\cS}{\mathcal{S}}
  \newcommand{\cT}{\mathcal{T}}
  \newcommand{\cU}{\mathcal{U}}
  \newcommand{\cV}{\mathcal{V}}
  \newcommand{\cW}{\mathcal{W}}
  \newcommand{\cX}{\mathcal{X}}
  \newcommand{\cY}{\mathcal{Y}}
  \newcommand{\cZ}{\mathcal{Z}}
  \newcommand{\rE}{\mathrm{E}}
  \newcommand{\rH}{\mathrm{H}}
  \newcommand{\rU}{\mathrm{U}}
  \newcommand{\Cp}{\mathrm{Cap\;}}
  \newcommand{\lan}{\langle}
  \newcommand{\ran}{\rangle}
  \newcommand{\an}[1]{\lan#1\ran}
  \def\diag{\mathop{{\rm diag}}\nolimits}
  \newcommand{\hs}{\hspace*{\parindent}}
  \newcommand{\cl}{\mathop{\mathrm{Cl}}\nolimits}
  \newcommand{\tr}{\mathop{\mathrm{Tr}}\nolimits}
  \newcommand{\Aut}{\mathop{\mathrm{Aut}}\nolimits}
  \newcommand{\argmax}{\mathop{\mathrm{arg\,max}}}
  \newcommand{\Eig}{\mathop{\mathrm{Eig}}\nolimits}
  \newcommand{\Gr}{\mathop{\mathrm{Gr}}\nolimits}
  \newcommand{\Fr}{\mathop{\mathrm{Fr}}\nolimits}
  \newcommand{\trans}{^\top}
  \newcommand{\opt}{\mathop{\mathrm{opt}}\nolimits}
  \newcommand{\per}{\mathop{\mathrm{perm}}\nolimits}
  \newcommand{\haff}{\mathrm{haf\;}}
  \newcommand{\perio}{\mathrm{per}}
  \newcommand{\conv}{\mathrm{conv\;}}
  \newcommand{\Cov}{\mathrm{Cov}}
  \newcommand{\inter}{\mathrm{int}}
  \newcommand{\dist}{\mathrm{dist}}
  \newcommand{\inn}{\mathrm{in}}
  \newcommand{\grank}{\mathrm{grank}}
  \newcommand{\mrank}{\mathrm{mrank}}
  \newcommand{\krank}{\mathrm{krank}}
  \newcommand{\out}{\mathrm{out}}
  \newcommand{\orient}{\mathrm{orient}}
  \newcommand{\Pu}{\mathrm{Pu}}
  \newcommand{\rdc}{\mathrm{rdc}}
  \newcommand{\range}{\mathrm{range\;}}
  \newcommand{\Sing}{\mathrm{Sing\;}}
  \newcommand{\topo}{\mathrm{top}}
  \newcommand{\undir}{\mathrm{undir}}
  \newcommand{\Var}{\mathrm{Var}}
  \newcommand{\rC}{\mathrm{C}}
  \newcommand{\rF}{\mathrm{F}}
  \newcommand{\rL}{\mathrm{L}}
  \newcommand{\rM}{\mathrm{M}}
  \newcommand{\rO}{\mathrm{O}}
  \newcommand{\rR}{\mathrm{R}}
  \newcommand{\rS}{\mathrm{S}}
  \newcommand{\rT}{\mathrm{T}}
  \newcommand{\pr}{\mathrm{pr}}
  \newcommand{\inte}{\mathrm{int}}
  \newcommand{\inv}{\mathrm{inv}}
  \newcommand{\pers}{\per_s}
  \newcommand{\del}{\boldsymbol{\delta}}
  \renewcommand{\alph}{\boldsymbol{\alpha}}
  \newcommand{\bet}{\boldsymbol{\beta}}
  \newcommand{\gam}{\boldsymbol{\gamma}}
  \newcommand{\sig}{\boldsymbol{\sigma}}
  \newcommand{\zet}{\boldsymbol{\zeta}}
  \newcommand{\et}{\boldsymbol{\eta}}
  \newcommand{\xit}{\boldsymbol{\xi}}
  \newcommand{\perm}{\mathrm{perm\;}}
  \newcommand{\adj}{\mathrm{adj\;}}
  \newcommand{\rank}{\mathrm{rank\;}}
  \newcommand{\set}[1]{\{#1\}}
  \newcommand{\spec}{\mathrm{spec\;}}
  \newcommand{\supp}{\mathrm{supp\;}}
  \newcommand{\Tr}{\mathrm{Tr\;}}
  \newcommand{\vol}{\text{vol}}

\theoremstyle{definition}
\newtheorem{definition}{Definition}

\theoremstyle{remark}
\newtheorem*{remark}{Remark}
\newtheorem{example}{Example}

\title{Closed formula for the relative entropy of entanglement in all dimensions}
\author{Shmuel Friedland}\email{friedlan@uic.edu}
\affiliation{Department of Mathematics, Statistics and Computer Science
University of Illinois at Chicago,
851 S. Morgan Street,
Chicago, IL 60607-7045}
\author{Gilad Gour}\email{gour@math.ucalgary.ca}
\affiliation{Institute for Quantum Information Science and
Department of Mathematics and Statistics,
University of Calgary, 2500 University Drive NW,
Calgary, Alberta, Canada T2N 1N4}

\date{1 October, 2010}

\begin{abstract}
The relative entropy of entanglement is defined in terms of the relative entropy between an entangled state and its
closest separable state (CSS).  Given a multipartite-state on the boundary of the set of separable states, we find a closed formula for \emph{all} the entangled states for which this state is a CSS. Our formula holds for multipartite states in all dimensions.
For the bipartite case of two qubits our formula reduce to the one given in Phys. Rev. A \textbf{78}, 032310 (2008).
\end{abstract}


\maketitle

\section{Introduction}

Immediately with the emergence of quantum information science (QIS), entanglement was recognized as the key resource for many tasks such as teleportation, super dense coding
and more recently measurement based quantum computation (for review, see e.g.~\cite{Hor09,Ple07}). This recognition sparked an enormous stream of work in an effort to quantify entanglement in both bipartite and multipartite settings.
Despite the huge effort, except the negativity~\cite{Vid02} (and the logarithmic negativity~\cite{Ple05}) closed formulas for the calculation of different measures of entanglement exist only in two qubits systems and,
to our knowledge, only for the entanglement of formation~\cite{Woo98}. Moreover, the discovery that several measures of entanglement and quantum channel capacities are not additive~\cite{Has09,Yard08}, made it clear that
formulas in lower dimensional systems, in general, can not be used to determine the asymptotic rates of different quantum information tasks. Hence, formulas in higher dimensional systems are quite essential for the development of QIS.

Among the different measures of entanglement, the relative entropy of entanglement (REE)  is of a particular importance. The REE
is defined by~\cite{Ved98}:
\begin{equation}\label{def}
E_{R}(\rho)=\min_{\sigma'\in\mathcal{D}}S(\rho\|\sigma')=S(\rho\|\sigma)\;,
\end{equation}
where $\mathcal{D}$ is the set of separable states or positive partial transpose (PPT) states, and
$S(\rho\|\sigma)\equiv\text{Tr}\left(\rho\log\rho-\rho\log\sigma\right)$. It
quantifies to what extent a given state can be operationally distinguished from the closest state
which is either separable or has a positive partial transpose (PPT). Besides of being an entanglement monotone it also has nice properties such as being asymptotically continuous. The importance of the REE comes from the fact that
its asymptotic version provides the unique rate for reversible transformations~\cite{Hor02}. This property was demonstrated recently with the discovery that the regularized REE is the unique function that quantify the rate of interconversion between states in a reversible theory of entanglement, where all types of non-entangling operations are allowed~\cite{Bra08}.

The state $\sigma=\sigma(\rho)$ in Eq.~(\ref{def}) is called the closest separable state (CSS) or the closest PPT state.
Recently, the inverse problem to the long standing problem~\cite{Eis05} of finding the formula for the CSS $\sigma(\rho)$ was solved in~\cite{MI} for the
case of two qubits. In~\cite{MI} the authors found a closed formula for the inverse problem. That is, for a given state $0<\sigma$
on the boundary of 2-qubits separable states, $\partial\mathcal{D}$, the authors found an explicit formula describing all entangled states for which $\sigma$ is the CSS. Quite astonishingly, we show here that this inverse problem can be solved analytically
not only for the case of two qubits, but in fact in all dimensions and for any number of parties.

We now describe briefly this formula. Denote by $\rH_n$ the Hilbert space of $n\times n$ hermitian
 matrices, where the inner product of $X,Y\in\rH_n$ is given by $\tr XY$.  Denote by $\rH_{n,+,1}\subset\rH_{n,+}\subset \rH_n$
 the convex set of positive hermitian matrices of trace one, and the cone of positive hermitian matrices, respectively.
 Here $n=n_1n_2\cdots n_s$ so that the multi-partite density matrix $\rho\in\rH_{n,+,1}$ can be viewed as acting on the $s$-parties Hilbert space $\mathbb{C}^{n_1}\otimes\mathbb{C}^{n_2}\otimes\cdots\otimes\mathbb{C}^{n_s}$.

Let $0<\sigma\in\rH_{n,+}$ (i.e. $\sigma$ is full rank).  Then for any $\sigma'\in\rH_n$ and
 a small real $\varepsilon$ we have the Taylor expansion of
 $$\log(\sigma+\varepsilon \sigma')=\log\sigma+\varepsilon L_{\sigma}(\sigma')+O(\varepsilon^2).$$
 Here $L_{\sigma}:\rH_n\to\rH_n$ is a self-adjoint operator (defined in the next section), which is invertible, and satisfies $L_{\sigma}(\sigma)=I$.

 Assume now that $0<\sigma\in\partial\mathcal{D}$ (later we will extend the results for all $\sigma\in\partial\mathcal{D}$; i.e. not necessarily full rank).  Then, from the supporting hyperplane theorem, $\sigma$ has at
 least one supporting hyperplane, $\phi\in\rH_n$, of the following form:
 \begin{equation}\label{phisuphy}
 \tr(\phi\sigma')\ge \tr(\phi\sigma)=0 \;\;\;\forall\;\;\; \sigma'\in \mathcal{D}\;,
 \end{equation}
 where $\phi$ is normalized; i.e. $\tr\phi^2=1$.
 For each such $\phi$, we define the family of \emph{all} entangled states, $\rho(x,\sigma)$, for which $\sigma$ is the CSS:
 \begin{equation}\label{rhoformul}
 \rho(x,\sigma)=\sigma - x L_{\sigma}^{-1}(\phi),\quad 0<x\le x_{\max}.
 \end{equation}
 Here, $x_{\max}$ is defined such that $\rho(x_{\max},\sigma)\in\rH_{n,+,1}$ and $\rho(x_{\max},\sigma)$
 has at least one zero eigenvalue. We also note that $\tr L_{\sigma}^{-1}(\phi)=0$. Moreover, for the case of two qubits,
 $\phi$ is unique and is given by $\phi=\left(|\varphi\rangle\langle\varphi|\right)^{\Gamma}$, where $\Gamma$ is the partial transpose, and $|\varphi\rangle$ is the unique normalized state that satisfies $\sigma^{\Gamma}|\varphi\rangle=0$.
 Hence, for the case of two qubits our formula is reduced to the one given in~\cite{MI}, by recognizing that for this case, the self-adjoint operator $L_{\sigma}^{-1}$ is given by the function $G(\sigma)$ of Ref.~\cite{MI}.

 This paper is organized as follows. In the next section we discuss the definition of $L_{\sigma}$.
 In section III we find necessary and sufficient conditions for the CSS and in section IV we prove the main result for the case were the CSS is full rank. To illustrate how the formula can be applied, in section V we discuss the qubit-qudit $2\times m$ case. In section VI we discuss how to apply the formula for tensor products. In section VII we discuss the singular case, and show that the CSS state can be described in a similar way to the non-singular case.  We end in section VIII with conclusions.

 \section{Definition of $L_{\sigma}$}

 Let $0<\alpha\in\rH_{n,+}$.  Fix $\beta\in \rH_n$.  Let $t\in (-\varepsilon,\varepsilon)$ for some small
 $\varepsilon=\varepsilon(\alpha)>0$.  Rellich's theorem, , e.g. \cite{Kat80}, yields that $\log(\alpha+t\beta)$ is analytic
 for $t\in (-\varepsilon,\varepsilon)$.  So
 \begin{equation}\label{analytic}
 \log(\alpha+t\beta)=\log \alpha+t\rL_{\alpha}(\beta)+O(t^2).
 \end{equation}
 Here $\rL_{\alpha}:\rH_n\to \rH_n$ is the following linear operator.
 In the eigenbasis of $\alpha$, $\alpha={\rm diag}(a_1,\ldots,a_n)$ is a diagonal matrix,
 where $a_1,\ldots,a_n>0$.  Then for $\beta=[b_{ij}]_{i,j=1}^n$ we have that
 \begin{equation}\label{frstvarlog}
 [\rL_{\alpha}(\beta)]_{kl} = b_{kl}\frac{\log a_k-\log a_{l}}{a_k-a_l}, \quad k,l=1,\ldots,n.
 \end{equation}
 Here we assume that for a positive $a$, $\frac{\log a-\log a}{a-a}=\frac{1}{a}, \frac{a-a}{\log a -\log a}=a$.

 Equivalently, for a real diagonal matrix $\alpha=(a_1,\ldots,a_n)>0$ define the real symmetric matrix
 \begin{align}
 &\left[T(\alpha)\right]_{k,l=1}^n=\frac{\log a_k-\log a_{l}}{a_k-a_l}\nonumber\\
 &\left[S(\alpha)\right]_{k,l=1}^n= \frac{a_k-a_l}{\log a_k-\log a_{l}}.\label{defTalpha}
 \end{align}
 Then, $L_{\alpha}(\beta)=\beta\circ T(\alpha)$, where
 $\beta\circ \eta$ is the entrywise product of two matrices, sometimes called the Hadamard product of matrices.
 Note that $L_{\alpha}$ is an invertible operator, where $L_{\alpha}^{-1}(\beta)=\beta\circ S(\alpha)$.

\section{A necessary and sufficient condition for $\sigma(\rho)$}

\begin{figure}[tp]
\includegraphics[scale=.40]{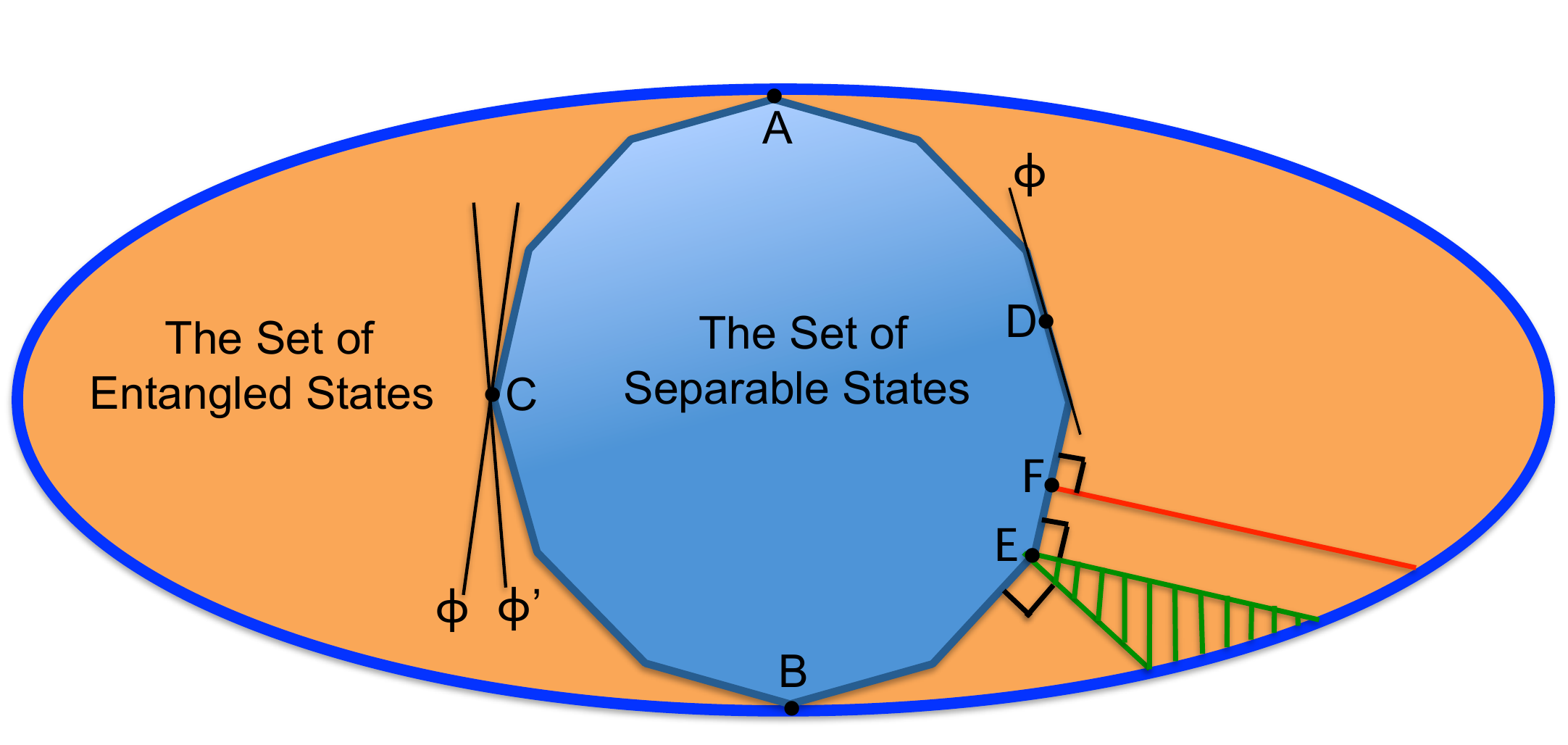}
\caption{A schematic diagram of separable states (blue) and entangled states (orange). Most points on the boundary, like the points D and F, have a unique supporting hyperplane (which is also the tangent plane). The point F is the CSS of all the points on the red line. Some of the points, like the points C and E, have more than one
supporting hyperplane. The point E is the CCS of all the points in the shaded green area. Some points on the boundary, like the points A and B, can not be a CSS; for example, separable states of rank 1 (i.e. product states) are on the boundary of separable states, but can never be the CSS of some entangled state.}
\end{figure}

We start with a necessary and sufficient condition the CSS, $\sigma(\rho)$, must satisfy.
\begin{theorem}\label{ness}
 Let $0<\rho\in\rH_{n,+,1}\backslash\mathcal{D}$.  The state $0<\sigma(\rho)\in\mathcal{D}$ is a solution to Eq.~(\ref{def}),
 if and only if $\sigma\equiv\sigma(\rho)$ satisfies
 \begin{equation}\label{necmaxcon}
 \max_{\sigma'\in \mathcal{D}} \tr  \sigma'L_{\sigma}(\rho)=\tr  \sigma L_{\sigma}(\rho)=1.
 \end{equation}
 \end{theorem}
 \begin{remark}
 We will see later that the assumptions that $0<\rho$ and $0<\sigma$ are not necessary.
 \end{remark}
 \begin{proof}
 First, note that $L_{\sigma}(\sigma)=I_n$, and $L_\sigma$ is a self-adjoint operator. Hence,
 $ \tr  \sigma L_{\sigma}(\rho)=\tr L_{\sigma}(\sigma)\rho=\tr (\rho)=1$.
 Now, let $\sigma'\in \mathcal{D}$.  Since $\mathcal{D}$ is a convex set, it follows that for every $t\in [0,1]$, $(1-t)\sigma+t\sigma'=\sigma+
 t(\sigma'-\sigma)\in\mathcal{D}$.  Thus, applying Rellich's theorem for a small $t>0$ gives
 \begin{equation}\label{nes1}
 \log(\sigma+t(\sigma'-\sigma))=\log\sigma+tL_{\sigma}
 (\sigma'-\sigma)+O(t^2).
 \end{equation}
If $\sigma$ is a solution to Eq.(\ref{def}), we must have $\tr \rho\log\sigma\ge \tr\rho\log\left[\sigma+t(\sigma'-\sigma)\right]$ which together with Eq.(\ref{nes1}) implies that for a small positive $t$,
 $t\tr \rho L_{\sigma}(\sigma'-\sigma)\le 0$.  Dividing by $t$ gives
 $\tr \rho L_{\sigma}(\sigma')\le \tr \rho L_{\sigma}(\sigma)=1$. This completes the necessary part of the proof since $L_\sigma$ is self-adjoint. The sufficient part of the proof follows directly from the construction of $\rho(x,\sigma)$ in Theorem~\ref{main}.
 \end{proof}

 The proposition above leads to the following intuitive corollary:
 \begin{corollary}
 Let $0<\sigma$ be a CSS of an entangled state $\rho$. Then, $\sigma\in\partial\mathcal{D}$.
  \end{corollary}
  \begin{proof}
Assume that the CSS $\sigma$ is an
interior point of $\mathcal{D}$.  Thus, for each $\sigma'$ separable
$(1-t)\sigma+t\sigma'$ is separable for all small $|t|$, where $t$ is either positive or
negative.
Hence, instead of Eq.(\ref{necmaxcon}) we get the identity
$\tr  L_{\sigma}(\rho)\sigma'=1$
for \emph{all} separable states $\sigma'\in\mathcal{D}$.  This yields that $L_{\sigma}(\rho)=I_n$.
Hence $\rho=\sigma$ which is impossible since $\rho$ was assumed not to be
separable.
  \end{proof}

 \section{Main Theorem}

 In the following we prove the main theorem for the case where the CSS, $\sigma$, is full rank.
 Note that if $\rho$ is full rank then $\sigma$ also must be full rank.

 \begin{theorem}\label{main}
\textbf{(a)} Let $0<\sigma\in\partial\mathcal{D}$, and let $\rho\in\rH_{n,+,1}$. Then,
 $$
 E_{R}(\rho)=S(\rho\|\sigma)\;\;\;\left(\text{i.e.}\;\sigma\;\text{is the CSS of}\;\rho\right)
 $$
 if and only if $\rho=\rho(x,\sigma)$, where $\rho(x,\sigma)$ is defined in Eq.(\ref{rhoformul}).\\
 \textbf{(b)} If $\rho>0$ (i.e. full rank) than the CSS is unique.
 \end{theorem}

 \begin{proof} \textbf{(a)}
 We first assume that $\sigma$ is a CSS of $\rho$.
 Recall that any linear functional $\Phi$ on $\rH_n$ is of the form $\Phi(X)=\tr(\phi X)$ for some $\phi\in\rH_n$.
 Now, since $\mathcal{D}$ is a closed convex subset of $\rH_{n,+,1}$ it follows (from the supporting hyperplane theorem) that
 for each boundary point $\sigma\in\partial\mathcal{D}$ there exists a nonzero linear functional on
 $\Phi:\rH_n\to \R$, represented by $\phi\in\rH_n$, satisfying the following condition:
 \begin{equation}
 \Phi(\sigma)\le \Phi(\sigma') \textrm{ for all } \sigma'\in\mathcal{D}.
 \end{equation}
 Note that the equation above holds true if $\phi$ is replaced by $\phi-aI$ (this is because $\tr\sigma=\tr\sigma'=1$).
 Moreover, since $\phi\neq 0$ we can normalize it. Therefore, there exists $\phi\in\rH_n$ satisfying (\ref{phisuphy})
 and the normalization
 \begin{equation}\label{norphi}
 \tr\phi^2=1.
 \end{equation}
 For most $\sigma$ on the boundary $\partial\mathcal{D}$, the supporting hyperplane of $\mathcal{D}$ at
 $\sigma$ is unique (see Fig.~1).  This is equivalent to say that $\phi\in\rH_n$ satisfying the conditions
 in Eq.~(\ref{phisuphy}) and (\ref{norphi})
 is unique.  However, for some special boundary points $\sigma\in\partial\mathcal{D}$,
 there is a cone of such $\phi$ of dimension greater than one satisfying (\ref{phisuphy}) (see Fig.~1).

 Now, denote $\phi':=-(L_{\sigma}(\rho)-I)$. Since we assume that $\sigma$ is a CSS of $\rho$ we get from
 Eq.~(\ref{necmaxcon}) the condition  $\tr(\phi'\sigma')\ge \tr(\phi'\sigma)=0$.  Recall that $L_{\sigma}(\sigma)=I$.
 Hence $\phi'=-L_{\sigma}(\rho-\sigma)$.  Since $\rho\ne \sigma$ and $L_{\sigma}$ is invertible, it follows that $\phi'\ne 0$.
 Hence, $\phi'$ can be normalized such that $\phi'=x\phi$, where $\phi$ satisfies Eq.~(\ref{norphi}) and $x>0$.  Apply $L_{\sigma}^{-1}$ to $\phi$ to deduce (\ref{rhoformul}).   We remark that $\tr L_{\sigma}^{-1}(\phi)=0$.  Indeed
 $$0=\tr \phi \sigma=\tr\phi L_{\sigma}^{-1}(I)=\tr L_{\sigma}^{-1}(\phi) I= \tr L_{\sigma}^{-1}(\phi).$$

Assume now that $0<\sigma\in\partial\mathcal{D}$, and let $\phi$ be a supporting hyperplane at $\sigma$, satisfying
Eq.~(\ref{phisuphy}) and (\ref{norphi}). Set $\rho\equiv\rho(x,\sigma)$ as in Eq.(\ref{rhoformul}). We want to show that for this $\rho$,
$E_R(\rho)=S(\rho\|\sigma)$.
 Recall first that the relative entropy $S(\rho\|\sigma'):=\tr(\rho\log\rho)-\tr(\rho\log\sigma')$ is jointly
 convex in its arguments~\cite[Thm 11.12]{NC}.  By fixing the first variable $\rho$ we deduce that
 $\tr (\rho\log\sigma')$ is concave on $\mathcal{D}$.  Consider the function
 $$f(t):=\tr (\rho\log((1-t)\sigma+t\sigma')), \quad t\in[0,1],$$
 where $\rho\ge 0$ is given by Eq.~(\ref{rhoformul}) and $x>0$.  The joint convexity of the relative entropy
 implies that $f(t)$ is concave.  Now, to show that the minimum of $S(\rho\|\sigma')$ is obtained at $\sigma'=\sigma$, it is
 enough to show that $f(0)\ge f(1)$ for each $\sigma'\in \mathcal{D}$.  To see that, we first show
 show that $f'(0)\le 0$, which then, combined with concavity of $f$, implies that $f(0)\ge f(1)$.
For small $t$ we have $$\log(\sigma+t(\sigma-\sigma'))=\log\sigma+tL_{\sigma}(\sigma'-\sigma)+O(t^2).$$
Hence,
 \begin{align}\label{compf'0}
f'(0) & =\tr(\rho L_{\sigma}(\sigma'-\sigma))=\nonumber\\
& =\tr\left[L_{\sigma}(\rho)(\sigma'-\sigma)\right]=\tr\left[(I-x\phi)(\sigma'-\sigma)\right]\nonumber\\
&=x\tr(\phi\sigma-\phi\sigma')=-x\tr(\phi\sigma')\le 0.
 \end{align}
 This completes the proof of part \textbf{(a)}. Moreover, $E_{R}(\rho)=\tr(\rho\log\rho)-\tr(\rho\log\sigma)>0$, since $\rho\ne \sigma$.  Hence $\rho$ is entangled.

 \textbf{(b)} This part follows from the strong concavity of $\log\sigma$ (see appendix~\ref{strong}).
 Therefore, from Corollary~\ref{conclog} of appendix A, it follows that for a fixed entangled state $\rho>0$, the function
 $\tr\rho\log\sigma$ is a strict concave function on the open set of all strictly positive Hermitian matrices in $\rH_n$.
 Hence, if both $\sigma$ and $\sigma'$ are CSS of $\rho$, then both are full rank and we have $\tr \rho\log\sigma=\tr\rho\log\sigma'$. Hence, for $t\in(0,1)$ we set $\sigma''\equiv t\sigma+(1-t)\sigma'$ and from the strong concavity
 $$
 \tr\rho\log\sigma''>t\tr\rho\log\sigma+(1-t)\tr\rho\log\sigma'=\tr\rho\log\sigma\;,
 $$
 in contradiction with the assumption that $\sigma$ is a CSS.
 \end{proof}

 \begin{corollary}
 Let $0<\rho\in\rH_{n,+,1}$ be entangled state and let $\sigma$ be a CSS of $\rho$.
 Then, $\sigma$ is also the CSS of $\rho(t)\equiv t\rho+(1-t)\sigma$ for all $t\in [0,t_{\max}]$,
 where $t_{\max}>1$ is the maximum $t$ such that $\rho(t)\geq 0$.
 \end{corollary}
 \begin{proof}
 Since $\sigma$ is the CSS of $\rho$, from theorem~\ref{main} we have $\rho=\rho(x,\sigma)$ for some $x$.
 Hence $\rho(t)=\rho(tx,\sigma)$ is of the same form. From theorem~\ref{main} $\sigma$ is a CSS of $\rho(t)$.
 \end{proof}
 \begin{remark}
A weaker version of the corollary above was proved in~\cite{Ved98}; note that here $t$ can be greater than one.
 \end{remark}

  \section{Bipartite partial transpose analysis}

 In the following we show how to apply Theorem~\ref{main} to specific examples. In particular, we focus
 on the bipartite case (i.e. $n=n_1n_2$) and we will assume that $\mathcal{D}$ in Eq.(\ref{def}) is the set of PPT states. In the $2\times2$ and $2\times3$ case, $\mathcal{D}$ is also the set of separable states~\cite{Hor95}. The boundary of the PPT states is simple to characterize. If $\sigma\in\mathcal{D}$ satisfies $\sigma>0$ and also $\sigma^{\Gamma}>0$, where $\Gamma$ is the partial transpose, then $\sigma$ must be an interior point of $\mathcal{D}$. If on the other hand $\sigma$ or $\sigma^{\Gamma}$ are singular, then $\sigma$ must be on the boundary of $\mathcal{D}$. We therefore have:
$$
\partial\mathcal{D}=\left\{\sigma\in\mathcal{D} \;\Big|\; det(\sigma^{\Gamma}\sigma)=0\right\}\;.
$$

Suppose now that $0<\sigma\in\partial\mathcal{D}$. Hence, $\sigma^{\Gamma}$ has at least one zero eigenvalue.
Let $|\varphi\rangle$ be a normalized eigenstate corresponding to an eigenvalue zero and define an Hermitian matrix $\phi=(|\varphi\rangle\langle\varphi|)^{\Gamma}$. Since the partial transpose is self-adjoint with respect to the inner product $\langle\rho,\rho'\rangle=\tr(\rho\rho')$, it follows that $\phi$ satisfies Eq.(\ref{phisuphy}) and is normalized (i.e. $\tr\phi^2=1$). That is,
$\phi$ represents the supporting hyperplane at $\sigma$.
Note that if $\sigma^{\Gamma}$ has more than one zero eigenvalue than clearly $\phi$ is not unique and in fact there is
a cone of supporting hyperplanes of $\mathcal{D}$ at $\sigma$ (see points C and E in Fig.1). To illustrate this point in
more details, we discuss now the case where $n_1=2$ (i.e. the first system is a qubit) and $n_2\equiv m$.

 In the $2\times m$ case, we can write any state $\sigma\in\rH_{2m,+,1}$  using the block representation of
 \begin{equation}\label{sigform}
 \sigma=\left[\begin{array}{cc}A&B\\\ B^\dag&C\end{array}\right]\in \C^{(2m)\times (2m)},
 \quad A,B,C\in\C^{m\times m},
 \end{equation}
 and $A^\dag=A$, $C^\dag=C$.
The partial transpose of $\sigma$ is given by (here the partial transpose corresponds to the transpose on the first qubit system; i.e.
it is the left partial transpose):
 $\sigma^{\Gamma}:=\left[\begin{array}{cc}A&B^\dag\\B&C\end{array}\right]$.
 The following theorem shows that $\sigma^{\Gamma}$ can have more then one zero eigenvalue.
 \begin{theorem}\label{pt2theo}  Let $m\ge 2$.  If $\sigma>0$ and $\sigma^{\Gamma}\ge 0$ then $\rank \sigma^{\Gamma}\ge m+1$.
 Furthermore, for each $k=0,\ldots,m-1$ there exist strictly positive
 hermitian matrices $\sigma\in\rH_{2m,+,1}$ such that $\sigma^{\Gamma}\ge 0,\rank \sigma^{\Gamma}=2m-k$.
 \end{theorem}
 \begin{proof}
 Recall that since $\sigma\in \rH_{2m,+,1}$ is strictly positive definite we have $A>0$.
 Hence, $\sigma$ and $\sigma^{\Gamma}$ are equivalent to the following
 block diagonal hermitian matrices
 \begin{align*}
 \hat \sigma & =\left[\begin{array}{cc}A&0\\0&C-B^\dag A^{-1}B\end{array}\right]\\
& =
 \left[\begin{array}{cc}I&0\\-B^\dag A^{-1}&I\end{array}\right]
 \left[\begin{array}{cc}A&B\\B^\dag&C\end{array}\right]
 \left[\begin{array}{cc}I&0\\-B^\dag A^{-1}&I\end{array}\right]^\dag\\
 \tilde \sigma & =\left[\begin{array}{cc}A&0\\0&C-BA^{-1}B^\dag\end{array}\right]\\
 & =
 \left[\begin{array}{cc}I&0\\-BA^{-1}&I\end{array}\right]
 \left[\begin{array}{cc}A&B^\dag\\B&C\end{array}\right]
 \left[\begin{array}{cc}I&0\\-BA^{-1}&I\end{array}\right]^\dag\;,
 \end{align*}
respectively.
Hence
 \begin{align*}
 & \sigma>0\iff C-B^\dag A^{-1}B>0\\
 & \sigma^{\Gamma}\ge 0\iff
 C-BA^{-1}B^\dag\ge 0
 \end{align*}
 Note first that $C\neq BA^{-1}B^\dag$. Otherwise, we get that $BA^{-1}B^\dag>B^\dag A^{-1}B$, and since
 $B^\dag A^{-1}B\ge 0$ it follows that each eigenvalue of  $BA^{-1}B^\dag$ must be positive and the $i-th$ eigenvalue of
 $BA^{-1}B^\dag$
 must be strictly greater then the $i-th$ eigenvalue of $B^\dag A^{-1}B$.
 This can not be true since $\det BA^{-1}B^\dag=\det B^\dag A^{-1}B$.
 Hence $\rank \sigma^{\Gamma}\geq m+1$. This complete the first part of the theorem.

 Next, let $E\ge 0$.   Then to satisfy the condition $\sigma^{\Gamma}\geq 0$ of the above inequality we define
 $C$ by
 \begin{equation}\label{defC22}
 C= BA^{-1}B^\dag+E,\;\Rightarrow\;  \rank \sigma^{\Gamma}= m+\rank E.
 \end{equation}
 With the above identity the condition that $\sigma>0$ is equivalent to
 \begin{equation}\label{Econd}
 BA^{-1}B^\dag-B^\dag A^{-1}B+E>0.
 \end{equation}

 We first show that one can choose $A>0$ and $B$ and $E\ge 0$
 such that $\rank E=1$ and Eq.~(\ref{Econd}) hold.  For this purpose, we will see that it is enough to find $A>0$
 and $B$ such that $ G:=BA^{-1}B^\dag- B^\dag A^{-1}B$ has exactly $m-1$
 strictly positive eigenvalues and one negative eigenvalue. 

 Let $F\ge 0$ given.  Then $F$ has the spectral decomposition $F=U\Lambda U^\dag$, where $U$ is unitary and
 $\Lambda\ge 0$ is a diagonal matrix with the diagonal entries equal to the nonnegative eigenvalues
 of $F$.  Choose  $B=U\Lambda^{\frac{1}{2}}$ and $A=I$.  Then $G=F-\Lambda$.
 We claim that we can choose $F$ such that $G$ has $m-1$ positive eigenvalues and one negative eigenvalue.

 Fix $H=[h_{ij}]\in \rH_{2m}$ with zero diagonal, i.e $h_{ii}=0$ for $i=1,\ldots,n$, and a diagonal $D=\diag(d_1,\ldots,d_m)$. Assume that
 $d_1>\ldots>d_m>0$.  Choose $t\gg 1$ and consider $H(t)=tD+H=t(D+\frac{1}{t}H)$.
 Set $z=\frac{1}{t}$ and recall that $D(z)=D+zH$ has analytic eigenvalues for small $z$.
 Since the eigenvalues of $D$ are simple, and $D\e_i=d_i\e_i$, where $\e_i=(\delta_{1i},\ldots,\delta_{ni})\trans$
 it follows that the eigenvalues $\lambda_1(z),\ldots,\lambda_m(z)$ of $D(z)$ have the Taylor expansion
 $$
 \lambda_i(z)=d_i+O(z^2), \quad i=1,\ldots,m
 $$
 since $H$ has zero diagonal, (see e.g.~\cite{Kat80}).
 By choosing $F=H(t), t\gg 1$ we deduce that $G(t):=H(t)-\Lambda(t)=H+O(\frac{1}{t})$.

 It is left to show that there exist hermitian $H$ with zero diagonal entries and $m-1$ positive eigenvalues. Let $\lambda_1\ge \ldots \ge \lambda_m$.  It is known (Schur's theorem, e.g.~\cite[(5.5.8)]{HJ99}) that the sequence $(\lambda_1,\ldots,\lambda_m)$
 must majorize the sequence of the diagonal entries $(0,\ldots,0)$ of $H$.
 \begin{equation}\label{majcon}
 \sum_{i=1}^r \lambda_i\ge \sum_{i=1}^r 0=0, r=1,\ldots,m-1, \;\; \sum_{i=1}^m \lambda_i=\sum_{i=1}^m 0 =0.
 \end{equation}
 Furthermore, if $\lambda_1\ge\ldots\ge \lambda_n$ satisfies the above conditions, then there exists a real
 symmetric matrix $H$ with zero diagonal and the eigenvalues $\lambda_1,\ldots,\lambda_m$ (see Theorem 4.3.32 in~\cite{HJ88}).
 Choose $\lambda_1\ge\ldots\ge\lambda_{m-1}>0$ and $\lambda_m=-\sum_{i=1}^{m-1} \lambda_i$.
 Then (\ref{majcon}) holds.  Thus there exists $H$ with zero diagonal and $m-1$ strictly positive eigenvalues.  Hence for $t\gg 1$ $G(t)$ has $m-1$ strictly positive eigenvalues. Choose $t_0\gg 1$
 and set $G=G(t_0)$.  Let $G|u\rangle=\lambda_m|u\rangle$, where $\lambda_m<0$.  Let $E_0=-2\lambda_m |u\rangle\langle u|$.  So $G+E_0>0$ and $\rank E_0=1$.  For $k>1$ let $E_1\ge 0$ such that $\rank (E_0+E_1)=k$.
 Then $E=E_0+E_1$.
 \end{proof}

 Note that from Theorem~\ref{main} it follows that
 we can rewrite the expression of the relative entropy of $\rho$
 similar to the formula (7) of Ref.~\cite{MI}. That is,
 \begin{equation}\label{relentrfor}
 E_R(\rho)=\tr(\rho\log\rho)-\tr(\sigma\log\sigma)+x\tr (L_{\sigma}^{-1}(\phi)\log\sigma).
 \end{equation}

 From the theorem above it follows that for the case $m=2$, if $\sigma>0$ then $\sigma^{\Gamma}$ can have at most
 one zero eigenvalue. Hence, for this case $\phi$ is unique as pointed out in~\cite{MI}. For $m=3$ it follows from the theorem
 that there exists $\sigma>0$ such that $\sigma^{\Gamma}$ has two independent eigenstates corresponding to zero eigenvalue.
 Here is an example of such a state $\sigma$ of the form (\ref{sigform})
 $$\sigma=\frac{1}{229}\left[\begin{array}{cccccc}1&0&0&0&6&8\\0&1&0&1&0&0\\0&0&1&0&0&0\\0&1&0&100&0&0\\
 6&0&0&0&46&60\\8&0&0&0&60&80
 \end{array}\right].$$

\section{Tensor Products}

 We now show briefly how to extend the results presented in this paper to tensor product of separable
 states. For this purpose we denote by $\mathcal{D}^{A}$ and $\mathcal{D}^{B}$ the set of separable states in Alice's lab and Bob's lab, respectively.  We also denote by $\mathcal{D}^{AB}$ the set of separable states of the composite system.
 First observe that if $\sigma_a\in \partial\mathcal{D}^A$ then for any separable state $\sigma_{b}'\in\mathcal{D}^{B}$
 the state $\sigma_a\otimes\sigma_{b}'\in \partial\mathcal{D}^{AB}$.  Furthermore, let $\phi_a\in\rH_{n}^{A}$ be a supporting hyperplane
 of $\mathcal{D}^{A}$ at $\sigma_a$ of the form given in Eq.~(\ref{phisuphy}) and (\ref{norphi}).
 Let $\phi_{b}'\in\rH_{n'}^{B}$, which is nonnegative on $\mathcal{D}^{B}$, i.e. $\tr(\phi_{b}'\sigma_{b}')\ge 0$ for all $\sigma_{b}'\in\mathcal{D}^{B}$ (i.e. $\phi_{b}'$ is an entanglement witness in Bob's lab).
 Assume the normalization $\tr((\phi_{b}')^2)=1$.
 Then it is straightforward to show that $\phi:=\phi_a\otimes\phi_{b}'$ satisfies Eq.~(\ref{phisuphy}) and (\ref{norphi})
 for any
 $\sigma'\in\mathcal{D}^{AB}$ and $\sigma=\sigma_a\otimes\sigma_{b}'$.

 Assume first that $\sigma_a>0,\sigma_{b}'>0$.  Then we can use $\phi=\phi_a\otimes \phi_{b}'$ in the formula
 of Eq.~(\ref{rhoformul})
 to find the corresponding entangled state $\rho\in \rH_{nn',+,1}$.  If $\sigma_a>0$ and $\sigma_{b}'$
 is singular, we can still use the formula in Eq.~(\ref{rhoformul}), where $\phi_{b}'\ge 0$ on $\mathcal{D}^B$ and $\phi_b'\x=\0$
 if $\sigma_{b}'\x=\0$.  If $\sigma_a$ is singular then we can use the formula given in the next section.

 \section{The case of singular CSS}\label{singular}

 If the entangled state $\rho$ is not full rank then the CSS $\sigma$ can be singular (i.e. not full rank).
 More precisely, if $\x$ is an eigenvector of $\sigma$ corresponding to zero eigenvalue then $\x$
 must also be an eigenvector of $\rho$ corresponding to zero eigenvector. For the singular $\sigma$
 we work below with the basis where $\sigma$ is diagonal
 \begin{equation}\label{singdiag}
 \sigma=\diag(s_1,\ldots,s_n)\;,
 \end{equation}
 where $s_1\ge\ldots \ge s_r> 0=s_{r+1}=\ldots=s_n$ and $1\le r<n$.
 Here $r=\rank \sigma<n$ since $\sigma$ is singular.
 Note that in this basis $\rho$ has the following block diagonal form
 \begin{equation}\label{bldiagf}
 \rho=\left[\begin{array}{cc} \rho_{11}&0\\0&0\end{array}\right], \textrm{where }  \rho_{11}\in \rH_{r,+,1}.
 \end{equation}

With this eigen-basis of $\sigma$, we define the matrices $T(\sigma),\;S(\sigma)$ on the support of $\sigma$ just as in Eq.~(\ref{defTalpha}),
and zero outside the support (i.e. the last $n-r$ rows and columns of $T(\alpha),S(\alpha)$ are set to zero).
 Note that with this definition
 $$T(\sigma)\circ S(\sigma)=S(\sigma)\circ T(\sigma)= P_{\sigma}=\diag(\underbrace{1,\ldots,1}_r,0,\ldots,0),$$
 where $P_{\sigma}$ the projection to the support of $\sigma$.  Define now the linear operators $L_{\sigma}, L_{\sigma}^{\ddag}:\rH_n\to\rH_n$
 \begin{equation}\label{defopMN}
 L_{\sigma}(\xi):=T(\sigma)\circ \xi, \quad L_{\sigma}^{\ddag}(\xi):=S(\sigma)\circ\xi.
 \end{equation}
 Then $L_{\sigma}$ and $L_{\sigma}^{\ddag}$ are selfadjoint and
 \begin{equation}\label{lsigrel}
 L_{\sigma}L_{\sigma}^{\ddag}=L_{\sigma}^{\ddag}L_{\sigma}=P_{\sigma}.
 \end{equation}
 Note that $L_{\sigma}^{\ddag}$ is the \emph{Moore-Penrose} inverse of $L_{\sigma}$, and that if $\sigma>0$ then $L_{\sigma}^{\ddag}=L_{\sigma}^{-1}$. Note also that $L_{\sigma}(\sigma)=P_{\sigma}$.

 With the above definition for $L_\sigma$, Eq.~(\ref{analytic}) can be generalized to the singular case:
 \begin{lemma}\label{sinexp}  Let $\sigma\in\rH_{n,+}$ be a nonzero singular matrix.  Let
 $\xi\in \rH_n$ be positive on the eigenvector subspace of $\sigma$
 corresponding to the zero eigenvalue.  ($\x^\dagger \xi \x> 0$ if $\sigma \x=\0$ and $\x\ne \0$.)
 Assume that $\rho\in \rH_n$ is nonzero and $\rho\x=0$ if $\sigma\x=0$.  Then there exists
 $\varepsilon>0$ such that for any $t\in (0,\varepsilon)$ the following hold.
 \begin{equation}\label{varfor}
\tr(\rho\log(\sigma+t\xi))=\tr(\rho\log\sigma)+t\tr(\rho L_{\sigma}(\xi))+O(t^2|\log t|).
 \end{equation}
 \end{lemma}
\begin{proof}
 Without a loss of generality we may assume that $\sigma$ and $\rho$ of the form (\ref{singdiag})
 and (\ref{bldiagf}).  (However we do not need the assumption that $\rho_{11}\ge 0$.)
 Then there exists $\varepsilon>0$ such that for $\sigma(t)>0$
 for $t\in(0,\varepsilon)$.  Rellich's theorem yields that the eigenvalues and the eigenvectors
 $\sigma(t):=\sigma+t\xi$ are analytic in $t$ for $|t|<\varepsilon$.  Let $s_1(t),\ldots,s_n(t)$ be the analytic eigenvalues of $\sigma(t)$ such that
 $$s_i(t)=s_i+b_it +\sum_{j=2}b_{ij}t^j,\quad i=1,\ldots,n.$$
 The positivity assumption on $\xi$ implies that $b_i>0$ for $i=r+1,\ldots,n$.
 (Note that $s_i=0$ for $i>r$.)
 Rellich's theorem also claims that the eigenvectors of $\sigma(t)$ can parameterized analytically.
 So there exists a unitary $U(t), t\in [0,\varepsilon)$, depending analytically on $t$, for $|t|<\varepsilon$, such that the following conditions hold.
\begin{align*}
 & \sigma(t)=U(t)\diag(s_1(t),\ldots,s_n(t))U(t)^\dagger\\
 & U(t)U^\dagger (t)=I_n,\;
 U(t)=\sum_{j=0}^{\infty} t^j U_j,\; U_0=I_n.
 \end{align*}
 Hence
 \begin{align*}
 &\log(\sigma+t\xi)  =U(t)\diag(\log s_1(t),\ldots,\log s_n(t))U^\dagger(t)\\
& =U(t)\diag(0,\ldots,0,\log s_{r+1}(t),\ldots \log (s_n(t)))U^\dagger (t)\\
& + U(t)\diag(\log s_1(t),\ldots,
 \log s_r(t),0,\ldots,0)U^\dagger (t).
 \end{align*}
 Note that the last term in this expression in analytic it $t$ for $|t|<\varepsilon$.
 Clearly, $\log s_i(t)= \log (b_it) $ + analytic term.  Observe next that
 \begin{align*}
 &U(t)\diag(0,\ldots,0,\log s_{r+1}(t),\ldots \log (s_n(t)))U^\dagger (t)\\
 &=\diag(0,\ldots,0,\log s_{r+1}(t),\ldots \log (s_n(t))) \\
&+ t U_1\diag(0,\ldots,0,\log s_{r+1}(t),\ldots \log (s_n(t))) \\
 &+t\diag(0,\ldots,0,\log s_{r+1}(t),\ldots \log (s_n(t)))U_1^\dagger +O(t^2|\log t|).
 \end{align*}
 Using the standard fact that $\tr XY=\tr YX$ and the form of $\rho$ given by (\ref{bldiagf}) we deduce that
\begin{align*}
&\tr (\rho U(t)\diag(0,\ldots,0,\log s_{r+1}(t),\ldots \log (s_n(t)))U^\dagger (t))\\
& = O(t^2|\log t|).
 \end{align*}
 Hence
 \begin{align*}
 &\tr(\rho\log(\sigma+t\xi))\\
 &=\tr(\rho U(t)\diag(\log s_1(t),\ldots,
 \log s_r(t),0,\ldots,0)U^\dagger (t))\\
& +O(t^2|\log t|).
 \end{align*}
 Similar expansion result hold when we replace $\sigma$ by a a diagonal $\alpha>0$ as in the beginning of this section.  Combine these results to deduce the validity of (\ref{varfor}).
 \end{proof}

 From the lemma above it follows that Proposition~\ref{ness} holds true also for singular $\rho$ and singular CSS $\sigma$.
 To see that, let $\sigma$ be singular
 CSS of an entangled state $\rho$, and suppose first that $\sigma'>0$. Define also $\sigma(t)\equiv (1-t)\sigma+t\sigma'=\sigma+t(\sigma'-\sigma)$. Note that $\xi:=\sigma'-\sigma$ satisfies the assumptions of Lemma~\ref{sinexp}.  Thus, the arguments in Proposition~\ref{ness} yield that
 $\tr \rho L_{\sigma}(\sigma')\le \tr \rho L_{\sigma}(\sigma)=1$.  Using the continuity
 argument, we deduce that this inequality hold for any $\sigma'\in\mathcal{D}$.

 In Eq.~(\ref{phisuphy}) we defined the supporting hyperplane in terms of Hermitian matrix $\phi$ satisfying that $\tr\sigma\phi=0$.
 As we will see below, it will be more convenient to represent the supporting hyperplane of $\mathcal{D}$ at $\sigma$ in terms of
 $\psi\equiv I-\phi$. That is, the supporting hyperplane will be described by
 the linear functional $\Psi:\rH_n\to \R$, defined by
 $\Psi(\xi)=\tr (\psi\xi)$, where $\psi$ satisfies:
 \begin{equation}\label{propsi}
 \tr(\psi\sigma')\leq\tr(\psi\sigma)=1\;,\;\text{for all}\;\sigma'\in\mathcal{D}.
 \end{equation}
Note also that $\tr P_{\sigma}\sigma'\leq 1$ and therefore for any $x\in[0,1]$, $\psi(x)\equiv x\psi+(1-x)P_{\sigma}$ also satisfies
the same condition $\tr(\psi(x)\sigma')\leq\tr(\psi(x)\sigma)=1$.
We now ready to prove the main theorem for the case of singular CSS.

 \begin{theorem}\label{mainsing}
Let $\sigma\in\partial\mathcal{D}$ be a singular matrix in the boundary of $\mathcal{D}$, and let $\rho\in\rH_{n,+,1}$ be an entangled state.
Then,
 $$
 E_{R}(\rho)=S(\rho\|\sigma)\;\;\;\left(\text{i.e.}\;\sigma\;\text{is the CSS of}\;\rho\right)
 $$
 if and only if $\rho$ is of the form
 \begin{equation}\label{rhosingxfrm}
 \rho(x,\sigma)=(1-x)\sigma+xL_{\sigma}^{\ddag}\left(\psi\right)\;,\;\;0<x_{\max}\leq 1.
 \end{equation}
 Here $x_{\max}\le 1$ is the maximum value of $x$ not greater than 1 such that $\rho(x,\sigma)\in\rH_{n,+,1}$.
 The supporting hyperplane of $\mathcal{D}$ at $\sigma$ is represented by $\psi\ne P_{\sigma}$, that satisfies the conditions in
 Eq.(\ref{propsi}) and is zero outside the support of $\sigma$ (i.e. if $\sigma\x=\0$ then $\psi\x=\0$).
 \end{theorem}

 \begin{proof}
 Suppose first that $\rho=\rho(x,\sigma)$. We want to prove that $\sigma$ is the CSS of $\rho$. First, observe that
$$
 \tr(L_{\sigma}^{\ddag}(\psi))=\tr(L_{\sigma}^{\ddag}(\psi)P_{\sigma})=\tr(\psi L_{\sigma}^{\ddag}(P_{\sigma}))=\tr(\psi\sigma)=1.
$$
 Hence $\rho(x,\sigma)\in\rH_{n,+,1}$ for $x\in[0,x_{\max}]$, where we assume that $x_{\max}\le 1$.
 It is left to show that
 $\tr (\rho(x,\sigma)\log\sigma')\le \tr (\rho(x,\sigma)\log\sigma)$ for any $\sigma'\in\mathcal{D}$.
 From the continuity argument, it is enough to show this inequality for all $\sigma'>0$.
 Let $\sigma(t)=(1-t)\sigma+\sigma't$.
 Let $f(t)=\tr(\rho(x,\sigma)\log\sigma(t))$.
 Using the equality (\ref{varfor}), similar to Eq.~(\ref{compf'0}), we get for $\rho=\rho(x,\sigma)$:
\begin{align}
 f'(0) & =\tr(\rho L_{\sigma}(\sigma'-\sigma))=\tr\left[(\sigma'-\sigma)L_{\sigma}(\rho)\right]\nonumber\\
 &=\tr\left[\big((1-x)P_{\sigma}+x\psi\big)(\sigma'-\sigma)\right]=\nonumber\\
 &=(1-x)\tr (\sigma' P_{\sigma})+x\tr (\sigma'\psi)-1\leq 0\;,
\end{align}
 where we have used that
 $L_{\sigma}(\rho)=(1-x)P_{\sigma}+x\psi$ and $x\in(0,1]$.  (Note that $\tr \sigma'P_{\sigma}\le \tr \sigma'=1$.)  Hence $f'(0)\le 0$, and since $f(t)$ is concave
 we have $f(0)\geq f(1)$,
 which implies that $\tr\left[\rho(x,\sigma)\log\sigma\right]\ge \tr\left[\rho(x,\sigma)\log\sigma'\right]$.
 This completes the second direction of the theorem.
 Moreover, note that
$E_{R}(\rho)=\tr(\rho\log\rho)-\tr(\rho\log\sigma)>0$, since $\rho\ne \sigma$.  Hence $\rho$ is
 entangled.

 Assume now that $\sigma$ is a singular CSS of an entangled state $\rho'$.
 Without a loss of generality we may assume that $\sigma$ and $\rho'$ of the form (\ref{singdiag})
 and (\ref{bldiagf}).
 Hence, from Proposition~\ref{ness}, when applied to the singular case
 (see the discussion above), we get
 $$
 \tr \sigma' L_{\sigma}(\rho')\le \tr\sigma L_{\sigma}(\rho')=1
 $$
 for all $\sigma'\in\mathcal{D}$.
 Denote $\psi'\equiv L_{\sigma}(\rho')$.  Note that $\psi'\neq P_{\sigma}$ since $\rho'\ne\sigma$.
 Hence, with this notations the equation above reads $\tr\psi'\sigma'\leq\tr\psi'\sigma=1$. Moreover, by definition $\psi'$ is zero outside the support of $\sigma$.   Then $\rho'=L_{\sigma}^{\ddagger}(L_{\sigma}(\rho'))=L_{\sigma}^{\ddagger}(\psi)$.  Hence Eq.~(\ref{rhosingxfrm})
 holds for $\rho=\rho', \psi=\psi'$ and $x=1$.
 Note that if we define $\psi$ by
 $$
 \psi=t\psi'+(1-t)P_{\sigma}
 $$
 for some $t\in(0,1]$, then $\psi$ also satisfies the requirements of the theorem. By taking $L_{\sigma}^{\ddag}$ on both sides of the equation above, we get $\rho=t\rho'+(1-t)\sigma$. This implies Eq.~(\ref{rhosingxfrm}) $x=t$.  Note that from the first part of the proof we indeed conclude that
 $\sigma$ is the CSS to $\rho$
 Hence  Eq.~(\ref{rhosingxfrm}) holds for any $x\in (0,1]$.
 This completes the proof of the theorem.
\end{proof}

 Recall that Theorem \ref{main} claimed that for $\rho>0$ the corresponding CSS is unique.
 This is no longer true if $\rho$ semi-positive definite \cite{MI}.  The reason for that is quite simple.
 \begin{theorem}\label{finitcond}
 Let $\rho\in \rH_{n,+,1}$ and assume that $\rho$ is singular.  Denote $f_{\rho}(\xi)=\tr \rho \log\xi$ for $\xi\in\rH_{n,+}$.
 Then $f_{\rho}(\xi)>-\infty$ for $\xi\in\rH_{n,+}$ if and only if one of the following condition holds.
 \begin{enumerate}
 \item\label{finitcond1}
 $\xi>0$.
 \item\label{finitcond2}  $\rho\x=\0$ if $\xi\x=\0$.
 Equivalently,  assume that $\rho\in\rH_{n,+}$ is in the block diagonal form (\ref{bldiagf}) where $\rho_{11}>0$.
 There exists a unitary matrix $U$ of order $n-r$ such that
 $\xi=\diag(I_r,U)^\dagger \diag(\xi_2,0)\diag(I_r,U)$, where $0<\xi_2\in\rH_{p,+}$ and $p\in[r,n-1]$.
 \end{enumerate}
 Denote by $\rH_{n,+}(\rho)$ the set of all $\xi\in\rH_{n,+}$ such that $f_{\rho}(\xi)>-\infty$.
 Then $\rH_{n,+}(\rho)$ is a convex set.  The function $f_{\rho}:\rH_{k,+}(\rho)\to \R$ is concave but
 not strictly concave.
 \end{theorem}
 \begin{proof}  Clearly, if $\xi>0$ then $f_{\rho}(\xi)>-\infty$.
 More precisely, Assume that $\x_1,\ldots,\x_n$ is an orthonormal system of eigenvectors of $\xi$ with the corresponding eigenvalues $x_1\ge\ldots
 \ge x_n>0$.  Then $f_{\rho}(\xi)=\sum_{i=1}^n (\log x_i) \x_i^{\dagger} \rho \x_i$.  Using the continuity argument we deduce that this formula
 remains valid for $\xi$ semipositive definite.  Hence $f_{\rho}(\xi)>-\infty$ if and only if $\x^{\dagger} \rho\x=0$ for each eigenvector $\x$ in the null  space of $\xi$.  Since $\rho\in\rH_{n,+}$ it follows that
 $\x^{\dagger}\rho\x=0\iff \rho\x=\0$.  This proves the first part of \ref{finitcond2}.  The second part of \ref{finitcond2} follows straightforward from this condition.

 We now show that $\rH_{n.+}(\rho)$ is a convex set. Let $\zeta\in \rH_{n,+}$ and $s>0$.  Then $g(s)=f_{\rho}(\zeta+sI_k)$ is a strictly increasing function
 on $(0,\infty)$.  (Choose an eigenbase of $\zeta$.)
 Assume that $\xi,\eta \in\rH_{n,+}(\rho)$.   Then for any $s>0,t\in (0,1)$, the concavity of $\log \zeta$ on $\zeta>0$ yields that
 $\log (t(\xi+sI_k)+(1-t)(\eta+sI_k))\ge t\log (\xi+sI_k)+(1-t)\log (\eta+sI_k),$
 which implies
 $f_{\rho} (t(\xi+sI_k)+(1-t)(\eta+sI_k))\ge tf_{\rho} (\xi+sI_k)+(1-t)f_{\rho} (\eta+sI_k)$.
 Letting $s\searrow 0$ and using the assumption that $f_{\rho} (\xi), f_{\rho}(\eta)>-\infty$ we deduce that $f_{\rho} (t\xi+(1-t)\eta)>-\infty$.
 Hence $\rH_{n,+}(\rho)$ is convex.  The above arguments show also that $f_{\rho}$ is a concave function on $\rH_{n.+}(\rho)$.

 It is left to show that $f_{\rho}$ is not strictly concave on $\rH_{n.+}(\rho)$.  Let $\xi=\diag(\beta,\xi_2), \eta=\diag(\beta,\eta_2)$, where $0<\beta\in \rH_{r,+}, \xi_2,\eta_2\in\rH_{n-r,+}$
 and $\xi_2\ne \eta_2$.  Clearly, $f_{\rho}(t\xi+(1-t)\eta)=\tr \rho_{11} \log \beta$ for all $t\in [0,1]$.
 \end{proof}
 \begin{corollary}   Let $\rho\in \rH_{n,+,1}$ and assume that $\rho$ is singular.   Then the set of CSS to $\rho$ is a compact convex
 on the boundary of $\cD$, which may contain more then one point.
 \end{corollary}
$\;$
\section{conclusions}

 To conclude, given a state $\sigma$ on the boundary of separable or PPT states, we have found a closed formula for \emph{all} entangled states for which $\sigma$ is a CSS. We have also shown that if $\sigma$ is full rank, than it is unique. Quite remarkably, our formula holds in all dimensions and for any number of parties. As an illustrating example, we have analyzed the case of qubit-qudit systems and described how to apply the formula for this case.

\emph{Acknowledgments:---}
GG research is supported by NSERC.

\pagebreak

\begin{appendix}

\section{The strong concavity of $\log A$}\label{strong}
 \begin{definition}\label{concaop} For an interval $\inter\subset \R$ let $\rH_n(\inter)$ be the set of all $n\times n$
 hermitian matrices whose eigenvalues are in $\inter$.  (Here $\inter$ can be open, closed, half open, half infinite or infinite.)
 Let $f:\rH_n(\inter)\to \rH_{n}$ be a continuous function.  $f$ is called \emph{monotone}, \emph{strict monotone} and
 \emph{strong monotone}  if for any  $C,A\in \rH_n(\inter)$ the corresponding conditions hold respectively:
 $C\ge A \Rightarrow f(C)\ge f(A)$,  $C> A\Rightarrow f(C)> f(A)$,
 $C\gneq A\Rightarrow f(C)\gneq f(A)$.
 $f$ is called \emph{concave},  \emph{strict concave} and  \emph{strong concave}
 if for any  $C,A\in \rH_n(\inter)$ the corresponding conditions hold respectively:
 $f((1-s)A+sB)\ge (1-s)f(A)+sf(B)$, $f((1-s)A+sB)> (1-s)f(A)+sf(B)$ if $\rank (A-B)=n$ and $s\in (0,1)$, $f((1-s)A+sB)\gneq (1-s)f(A)+sf(B)$
 if $A\ne B$ and $s\in (0,1)$.
 \end{definition}

 A well known result is that the functions $f_t(A):=A^t, t\in (0,1)$ and $\log A$ are strictly concave and strictly monotone on $H_{n}((0,\infty))$.  See \cite[\S6.6]{HJ99}.  In this section we show that $\log A$ is strongly concave on  $H_{n}((0,\infty))$.  This implies that $\tr \rho\log\sigma$
 is strictly concave for a fixed $\rho>0$ and all $\sigma>0$.  Hence the CSS $\sigma$ to an entangled $\rho>0$  is strictly positive and unique.
 We need also to consider $\x^{\dagger}(\log A)\x$, where $\x$ is a nonzero column vector in $\C^n$, and $A$ is singular and positive.  Then it makes sense only to consider only those $\x\in\U_+(A)\subset\C^n$, where
 $\U_+(A)$ is  the subspace spanned by eigenvectors of $A$ corresponding to positive eigenvalues.  For $\x\in\U_+(A)$ we have that $\x^\dagger(\log A) \x>-\infty$.  We also agree that for each $\x\in\C^n\backslash \U_+(A)$
 $\x^\dagger(\log A) \x=-\infty$.  Then for each $\x\in\C^n$, the function $\x^\dagger(\log A)\x$ is concave on $\rH_{n,+}$.  we agree here that $$t(-\infty)=-\infty=-\infty-\infty=-\infty+\R=\R-\infty \textrm{ for any } t>0.$$
 \begin{theorem}\label{strongconcav}  Let $A,B\in \rH_{n,+}$.  Then
 \begin{equation}\label{RAB}
 R(A,B):=A+B-(A^{\frac{1}{2}}B^{\frac{1}{2}}+B^{\frac{1}{2}}A^{\frac{1}{2}})\ge 0.
 \end{equation}
 Furthermore one has the identity
 \begin{equation}\label{basid}
 ((1-s)A^{\frac{1}{2}}+sB^{\frac{1}{2}})^2=(1-s)A+sB-(1-s)sR(A,B)).
 \end{equation}
 Hence for $t\in (0,\frac{1}{2}]$
 \begin{align}
 & (1-s)A^{t}+sB^{t} \le ((1-s)A+sB\nonumber\\
 & -(1-s)s R(A,B))^{t}
 \le((1-s)A+sB)^{t},\label{strongtconc}\\
 & (1-s)\log A+s\log B  \le \log\Big[(1-s)A+sB\nonumber\\
 &-(1-s)s R(A,B)\Big]\le \log ((1-s)A+sB).
 \label{stronglogconcav1}
 \end{align}
 \end{theorem}
 \proof
 Let $A,B\in\rH_{n,+}$ and assume that $R(A,B)$ is defined by (\ref{RAB}).  We claim that $R(A,B)\ge 0$.
 This is a straightforward consequence of
 the Cauchy-Schwarz and the arithmetic-geometric inequalities
 \begin{eqnarray*}
 |\x^\dagger A^{\frac{1}{2}}B^{\frac{1}{2}}\x|\le ((\x^\dagger A\x)(\x^\dagger B\x))^{\frac{1}{2}}\le \frac{1}{2}(\x^\dagger A\x+\x^\dagger B\x),\\
 |\x^\dagger B^{\frac{1}{2}}A^{\frac{1}{2}}\x|\le ((\x^\dagger B\x)(\x^\dagger A\x))^{\frac{1}{2}}\le \frac{1}{2}(\x^\dagger B\x+\x^\dagger A\x).
 \end{eqnarray*}
 Furthermore, $R(A,B)=0$ if and only if $A=B$.
 Clearly $R(A,A)=0$.  Suppose that $R(A,B)=0$.
 Then the above arguments yield that we must have the equalities in the Cauchy-Schwarz
 inequalities, and equalities in the arithmetic-geometric mean for each $\x$.  So $A^{\frac{1}{2}}\x=B^{\frac{1}{2}}\x$ for each $\x$.
 Hence $A^{\frac{1}{2}}=B^{\frac{1}{2}}\Rightarrow A=B$.

 A straightforward calculation shows the validity of (\ref{basid}).  Hence
 \begin{align*}
 \left((1-s)A^{\frac{1}{2}}+sB^{\frac{1}{2}}\right)^2 & =(1-s)A+sB-(1-s)sR(A,B))\\
& \le (1-s)A+sB \textrm{ for } s\in (0,1).
 \end{align*}
 Since $A^{\frac{1}{2}}$ is monotone, we deduce from the above inequality the inequality
 (\ref{strongtconc}) for $t=\frac{1}{2}$.
 The inequality  (\ref{strongtconc}) for $t=\frac{1}{2^m}$ follows by induction.  Hence (\ref{stronglogconcav1}) follows
 from (\ref{strongtconc}) for $t=\frac{1}{2^m}$.

 Assume that $t\in (0,\frac{1}{2})$.  By assuming that $A^{2t}$ is concave and order preserving we deduce from the above inequality
 (\ref{strongtconc}) for $t\in(0,\frac{1}{2}]$.  \qed
 \begin{corollary}\label{conclog}  For each $\rho\in\rH_{n}((0,\infty))$ the function $\tr(\rho\log\sigma)$ is a strict concave function on
 $\rH_{n}((0,\infty))$.
 \end{corollary}
 \begin{proof}  Let $\sigma,\eta\in \rH_{n}((0,\infty))$ and assume that $\sigma\ne\eta$.  Then $R(\eta,\sigma)\gneq 0$.
 Hence for $s\in (0,1)$
 \begin{align*}
 & \log((1-s)\sigma +s\eta-(1-s)s R(\sigma,\eta))\lneq \log((1-s)\sigma +s\eta)\\
 & \Rightarrow\;\;
 (1-s)\log\sigma+s\log\eta\lneq \log((1-s)\sigma +s\eta)\\
 & \Rightarrow
 \tr(\rho((1-s)\log\sigma+s\log\eta))< \tr(\rho\log((1-s)\sigma +s\eta)).
 \end{align*}
 \end{proof}
 Obviously the above corollary does not hold if $\rho\ge 0$ has at least one zero eigenvalue.
 If $\sigma,\eta>0$ has same eigenvalues and eigenvectors which span the range of $\rho$, i.e. $\sigma \U_+(\rho)=\eta\U_+(\rho)=\U_+(\rho),
 \sigma=\eta|\U_+(\rho)$,
 then $\tr(\rho((1-s)\log\sigma+s\log\eta))$ is constant for $s\in [0,1]$.

 \end{appendix}

\end{document}